\documentclass[3p, sort&compress]{elsarticle}
\usepackage{amsmath}
\usepackage{color}
\usepackage{amsthm}
\usepackage{amssymb}
\usepackage{titlesec}
\usepackage[thinlines,thiklines]{easybmat}

\usepackage{graphicx}
\usepackage{epsfig}
\theoremstyle{definition}\newtheorem{Df}{Definition}
\theoremstyle{plain}\newtheorem{Th}{Theorem}
\theoremstyle{definition}\newtheorem{Rm}{Remark}
\theoremstyle{definition}
\theoremstyle{plain}
\theoremstyle{plain}
\theoremstyle{plain}\newtheorem{Lm}{Lemma}
\theoremstyle{plain}

\begin{document}
\begin{frontmatter}
\title{Generalizations of the distributed Deutsch-Jozsa promise problem}

\author{Jozef Gruska$^{1}$}
\author{Daowen Qiu$^{2}$ }
\author{Shenggen Zheng$^{1,}$\corref{one}}

 \cortext[one]{Corresponding author.\\ \indent{\it E-mail addresses:} zhengshenggen@gmail.com (S. Zheng), gruska@fi.muni.cz (J. Gruska),  issqdw@mail.sysu.edu.cn (D. Qiu).}

\address{

  $^{1}$Faculty of Informatics, Masaryk University, Brno 60200, Czech Republic\\

  $^2$Department of Computer Science, Sun Yat-sen University,
Guangzhou 510006,
  China\\
}

\begin{abstract}

In the {\em  distributed  Deutsch-Jozsa promise problem}, two parties are to determine whether their respective strings $x,y\in\{0,1\}^n$ are at the {\em Hamming distance} $H(x,y)=0$  or $H(x,y)=\frac{n}{2}$.
Buhrman et al. (STOC' 98) proved that the exact {\em quantum communication complexity} of this problem is  ${\bf O}(\log {n})$  while the {\em  deterministic  communication complexity} is ${\bf \Omega}(n)$. This was the first impressive  (exponential) gap between quantum and classical communication complexity.
 In this paper, we generalize the above  distributed Deutsch-Jozsa promise problem to determine, for any fixed  $\frac{n}{2}\leq k\leq n$, whether $H(x,y)=0$  or $H(x,y)= k$,  and  show that an exponential gap between exact quantum and deterministic communication complexity still holds if $k$ is an even such that $\frac{1}{2}n\leq k<(1-\lambda) n$, where $0< \lambda<\frac{1}{2}$ is given. We  also deal with a promise version of the well-known {\em disjointness} problem and  show  also that for this promise problem there exists an exponential gap between  quantum (and also probabilistic) communication complexity  and  deterministic communication complexity of the promise version of such a disjointness problem. Finally, some applications to  quantum, probabilistic and deterministic finite automata of the results obtained are demonstrated.

\end{abstract}

\begin{keyword}
Quantum communication complexity\sep  Deutsch-Jozsa promise problem \sep Query complexity  \sep Finite automata

\end{keyword}

\end{frontmatter}

\section{Introduction}

Since the topic of communication complexity was introduced by Yao  \cite{Yao79}, it has been  extensively studied \cite{Bra03,Buh09,Hro97,KusNis97}.
In the setting of two  parties, Alice is given  an $x\in\{0,1\}^n$, Bob is given a $y\in\{0,1\}^n$ and their task is to communicate in order to determine the value of some given Boolean function $f:\{0,1\}^n\times\{0,1\}^n\to\{0,1\}$, while exchanging as small  number of bits  as possible. In this setting, local computations of the parties
are considered to be free, but communication is considered to be  expensive and has to be minimized. Moreover,
for computation, Alice and Bob  have access to arbitrary computational power.

There are usually three types of communication complexities considered for the above communication task: deterministic, probabilistic or quantum.

Two of the most often studied communication problems  are  that of equality and disjointness \cite{KusNis97}, defined as follows:
\begin{itemize}
\item {\bf Equality}: $\text{EQ}(x,y)=1$ if $x=y$ and 0 otherwise.
  \item {\bf Disjointness}: $\text{DISJ}(x,y)=1$  if there is no index $i$ such that $x_i=y_i=1$ and $0$ if such an index exists.   Equivalently, this function can be defined also as
  $\text{DISJ}(x,y)=1$ if $\sum_{i=1}^n x_i\wedge y_i=0$ and  $0$ if $\sum_{i=1}^n x_i\wedge y_i>0$. (We can  view $x$ and $y$ as being subsets of $\{1,\cdots,n\}$ represented by  characteristic vectors and to have $\text{DISJ}(x,y)=1$ iff these two subsets are disjoint.)
\end{itemize}
Deterministic  communication complexities of   the above problems $\text{EQ}$ and $\text{DISJ}$  are both $n$ \cite{KusNis97}.

Buhrman et al. \cite{Buh98,Buh09} proved that the exact quantum communication complexity of the distributed Deutsch-Jozsa promise problem, for $x,y\in\{0,1\}^n$ and $n$ is even, that is for
 \begin{equation}
\text{EQ}'(x,y)=\left\{\begin{array}{ll}
                    1 &\ \text{if}\ H(x,y)=0 \\
                    0 &\ \text{if}\ H(x,y)=\frac{n}{2},
                  \end{array}
 \right.
\end{equation}
is  ${\bf O}(\log {n})$.
This was the first impressively large (exponential) gap between quantum and classical communication complexity\footnote{In fact,  both $n$ and $\frac{n}{2}$ must be even in order to obtain an exponential quantum speed-up. We will justify this claim  in the Remark \ref{rm-odd-k} in Section \ref{section3}.}.

It has been so far a folklore belief that  the promise $H(x,y)=\frac{n}{2}$ is essential for the above result. However, we prove that the result holds also for  the following generalizations of this promise problem
\begin{equation}
\text{EQ}_k(x,y)=\left\{\begin{array}{ll}
                    1 &\ \text{if}\ H(x,y)=0 \\
                    0 &\ \text{if}\ H(x,y)=k,
                  \end{array}
 \right.
\end{equation}
for any fixed $k\geq \frac{n}{2}$.  That is the exact quantum communication complexity of $\text{EQ}_k$ is  ${\bf O}(\log {n})$  while the classical deterministic  communication complexity is ${\bf \Omega}(n)$ if $k$ is an even such that  $\frac{1}{2}n\leq k<(1-\lambda) n$, where $0< \lambda<\frac{1}{2}$ is given. Our proof has been  inspired by   methods used in \cite{Amb13}.

Let us  consider also the following   problem. Namely, an
  analogue of the
 Deutsch-Jozsa promise problem:
\begin{equation}
\text{DJ}_k(x)=\left\{\begin{array}{ll}
                    1 &\ \text{if}\ W(x)=0 \\
                    0 &\ \text{if}\ W(x)>k,
                  \end{array}
 \right.
 \end{equation}
 where  $k\geq \frac{n}{2}$ is fixed and $W(x)$ is the Hamming weight of $x$. We prove that the exact quantum query complexity of $\text{DJ}_k$ is 1 while the  deterministic query complexity is $n-k+1$.

  If errors can be  tolerated,  both quantum and probabilistic communication complexities of the  equality problem are ${\bf O}(\log {n})$.

Concerning   disjointness problem,  the probabilistic communication complexity  is ${\bf \Omega}(n)$ \cite{Bar02,Ks92,Raz92}  even if errors are tolerated.
   In the quantum cases, Buhrman et al.~\cite{Buh98} proved that  quantum communication complexity of $\text{DISJ}$ is ${\bf O}(\sqrt{n}\log n)$. This bound has been improved to ${\bf O}(\sqrt{n})$ by Aaronson and Ambainis \cite{AA03}. Finally, Razborov
showed that any bounded-error  quantum  protocol for $\text{DISJ}$  needs to communicate about $\sqrt{n}$ qubits \cite{Raz03}.
Situation is different from   the $\text{EQ}$ problem,  for which there is an exponential gap between quantum (and also probabilistic) communication complexity  and deterministic communication complexity  as shown in \cite{KusNis97,Buh98,Buh09}.  All  known gaps for $\text{DISJ}$ are not larger than quadratic.
It is therefore of interest to find out whether there are some promise versions of the  disjointness problem for which bigger communication complexity gaps can be obtained.  We give a positive answer to such a question. In order to do that,
 we consider the following set of promise problems where   $0<\lambda\leq \frac{1}{4}$
 \begin{equation}
    \text{DISJ}_{\lambda}(x,y)=\left\{\begin{array}{ll}
                    1 &\ \text{if}\  \sum_{i=1}^n x_i\wedge y_i=0 \\
                    0 &\ \text{if}\ \lambda n\leq \sum_{i=1}^n x_i\wedge y_i\leq (1-\lambda) n.
                  \end{array}
 \right.
\end{equation}

 We  prove that  quantum communication complexity of $\text{DISJ}_{\lambda}$ is not more than $\frac{\log 3}{3\lambda}(3+2\log n)$ while the deterministic communication complexity is  ${\bf \Omega}(n)$.
  For example, if $\lambda=\frac{1}{4}$, then the quantum communication complexity of $\text{DISJ}_{\lambda}$ is not more than $3+2{\log n}$ while the  deterministic communication complexity is  more than 0.007n. We prove also that probabilistic communication complexity of  $\text{DISJ}_{\lambda}$ is not more than $\frac{\log{3}}{\lambda}\log{n}$. Therefore, there is an exponential gap between quantum (and also probabilistic) communication complexity  and deterministic communication complexity of the above promise problem.

Number of states is a natural complexity measure for all models of finite automata and  state complexity of finite automata is one of the  research fields with many applications \cite{Yu05}.
There is a variety of methods how to prove lower bounds on the state complexity  and methods as well as the results of communication complexity are among the main ones \cite{Hro01,Kla00,KusNis97}. In this paper we also show how to make use of our new communication complexity results to get  new state complexity bounds.

 The paper is structured as follows. In Section 2 basic needed  concepts and notations are introduced and  models   involved are
described in  details.  Communication complexities and query complexities of the promise problems  $\text{EQ}_{k}$ and $\text{DJ}_{k}$ are investigated in Section 3. Communication complexity of the promise problem  $\text{DISJ}_{\lambda}$   is dealt with in Section 4.  Applications to finite automata are explored in Section 5. Some open
problems are discussed in Section 6.

\section{Preliminaries}

In this section, we recall some  basic definitions about communication complexity, query complexity and quantum finite automata. Concerning basic concepts and notations of quantum information processing, we refer the reader to \cite{Gru99,Nie00}.

\subsection{Communication complexity}

We recall here only  very  basic concepts and notations of {\it communication complexity},
and we refer the reader to \cite{Buh09,KusNis97} for
more details. We will deal with the situation that there are  two
communicating  parties and with very simple tasks of
computing two-argument Boolean functions for the case  one argument is known to
one party and the other argument is known to the other party.
We will completely ignore computational resources needed by
parties and  focus solely on the amount of communication that is need to be
exchanged between both parties in order to compute the value of a given  Boolean function.

More technically, let $X=Y=\{0,1\}^n$. We will consider  two-argument functions $f: X\times Y\rightarrow \{0,1\}$ and two communicating parties.  Alice will be given an $x\in X$ and Bob   a $y\in Y$. They want to compute $f(x,y)$. If $f$ is  defined only on a proper subset of $X\times Y$,  $f$ is said to be a partial function or a promise problem.

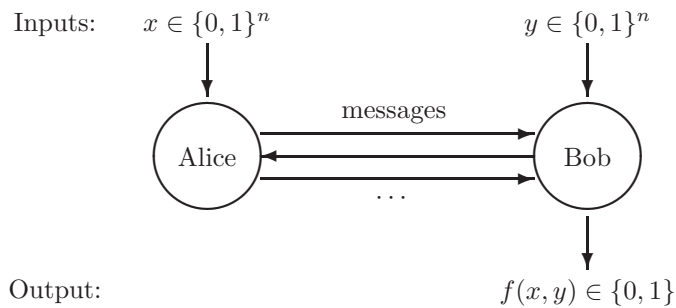
\begin{figure}[htbp]
  \centering
  \setlength{\unitlength}{1cm}
\begin{picture}(60,4)\thicklines

\put(6,3.5){\vector(0,-1){0.7}\makebox(0,0.5){$x\in\{0,1\}^n$}\makebox(-4,0.5){Inputs:}}
\put(11,3.5){\vector(0,-1){0.7}\makebox(0,0.5){$y\in\{0,1\}^n$}}

\put(11,1.2){\vector(0,-1){0.7}\makebox(0,-2){$f(x,y)\in\{0,1\}$}\makebox(-14,-2){Output:}}

\put(6,2){\circle{2}\makebox(0,0){Alice}}

\put(11,2){\circle{2}\makebox(0,0){Bob}}

\put(6.7,2.3){\vector(1,0){3.6}\makebox(-3.7,0.5){messages}}
\put(10.3,2){\vector(-1,0){3.6}\makebox(-3.7,0.5)}
\put(6.7,1.7){\vector(1,0){3.6}\makebox(-3.7,-0.5){$\cdots$}}
\end{picture}
  \centering\caption{Communication protocol}\label{f-com-pro}
\end{figure}

The computation of   $f(x,y)$ will be  done using a communication protocol, presented in Figure \ref{f-com-pro}. During the
execution of the protocol,  parties alternate roles in sending messages. Each of these
messages will be a bit-string. The protocol, whose steps are based on the communication so far, also specifies  for each step whether the communication terminates
(in which case it also specifies what is the output). If the communication does
not  terminate, the protocol also specifies what kind of  message the sender (Alice or Bob) should send next as
a function of its input and communication so far.

A deterministic communication protocol ${\cal P}$ computes
a (partial) function $f$, if for every (promise) input pair $(x,y)\in X\times Y$ the protocol terminates with the
value $f(x,y)$ as its output.
In a probabilistic  protocol, Alice and Bob may also flip coins during the protocol execution and proceed according to their outputs and the protocol can also have an erroneous output with a small probability.
In a  quantum protocol, Alice and Bob may   use also quantum resources   for communication.

 Let ${\cal P}(x,y)$ denote the  output of the protocol ${\cal P}$. We will consider two  kinds of protocols for computing a function $f$:
\begin{itemize}
  \item An exact protocol, that always outputs the correct answer (that is $Pr({\cal P}(x,y)=f(x,y))=1$).
  \item A  two-sided error (bounded error) protocol ${\cal P}$  such that  $Pr({\cal P}(x,y)=f(x,y))\geq \frac{2}{3}$.
\end{itemize}

The communication complexity of a protocol ${\cal P}$  is the
worst case number of (qu)bits exchanged.  The communication complexity of $f$ is, with  which respect to the communication mode used,  the complexity of an
optimal protocol for $f$.

We will use $D(f)$ and $R(f)$ to denote the deterministic communication complexity and the two-sided error probabilistic communication complexity of a function $f$, respectively. Similarly we use notations $Q_E(f)$ and $Q(f)$ for the exact  and two-sided error quantum communication complexity of a function $f$.

Let us  also summarize  already known communication complexity results concerning communication problems  $\text{EQ}$,  $\text{DISJ}$ and  $\text{EQ}'$:
\begin{enumerate}
  \item [1.] $D(\text{EQ})=n$, $D(\text{DISJ})=n$  \cite{KusNis97}, $D(\text{EQ}')\in {\bf \Omega}(n)$ \cite{Buh98}.
  \item [2.] $Q_E(\text{EQ}')\in {\bf O}(\log n)$ \cite{Buh98}.

  \item [3.]$R(\text{EQ})\in {\bf O}(\log n)$ \cite{KusNis97}, $R(\text{DISJ})\in {\bf \Omega}(n)$ \cite{Bar02,Ks92,Raz92}.
  \item [4.]$Q(\text{DISJ})\in {\bf \Theta}(\sqrt{n})$ \cite{AA03,Raz03}.
\end{enumerate}

\subsection{Exact query complexity}

The exact quantum query complexity for partial functions  was dealt with also in \cite{BH97,DJ92} and for total functions  in \cite{Amb13,AISJ13,AGZ14,MJM11}.

In the next we recall  definitions of two exact query complexity models.   For  more concerning basic concepts and notations related to  query complexity, we refer the reader to \cite{BdW02}.

Exact classical (deterministic) query algorithms  to compute a Boolean function $f:\{0,1\}^n\to \{0,1\}$ can be described
using  decision trees, in the following way:

Let the input string be $x=x_1x_2\ldots x_n$. A decision tree $T_f$ for $x$ is a rooted binary tree in which each internal vertex
has exactly two children. Moreover, each internal vertex is labeled with a variable $x_i$ ($1\leq i\leq n$) and each leaf is labeled with a value 0 or 1. $T_f$ should be designed in such a way that it can be used to compute function $f$ in the following way: Let us start at the root. If this is a leaf then stop and the value of $f$ is that assigned to that leaf. Otherwise, query the value of the variable $x_i$ that labels the root. If $x_i=0$, then evaluate recursively  the left subtree, if $x_i=1$ then  the right subtree. The output of the tree is then the value of the leaf that is reached eventually.
The depth
of $T_f$ is the maximal length of any path from the root to any leaf (i.e. the worst-case number of queries
used for all inputs). The minimal depth over all decision trees computing $f$ is the exact classical query complexity (deterministic
query complexity, decision tree complexity) $DT(f)$ of $f$.

Let $f:\{0,1\}^n\to \{0,1\}$ be a Boolean function and $x = x_1x_2\cdots x_n$ be an input
bit string.   Each exact quantum query algorithm for $f$
works in a Hilbert space with some fixed basis, called standard. Each of the basis states corresponds to either one or none
of the input bits. It starts in a
fixed starting state, then performs on it a sequence of  transformations
$U_1$, $Q$, $U_2$, $Q$, \ldots, $U_t$, $Q$, $U_{t+1}$.
Unitary transformations $U_i$ do not depend on
the input bits, while $Q$, called the {\em query transformation}, does,
in the following way.  If a basis state $|\psi\rangle$ corresponds to the $i$-th
input bit, then $Q|\psi\rangle=(-1)^{x_i}|\psi\rangle$. If it does not correspond to any
input bit, then $Q$ leaves it unchanged: $Q|\psi\rangle=|\psi\rangle$. Finally, the algorithm performs a  measurement in the standard basis.
Depending on the result of the measurement, the algorithm outputs either 0 or 1
which must be equal to $f(x)$. The {\em  exact quantum query complexity}
$QT_E(f)$ is the minimum number of queries used by any quantum algorithm which
computes $f(x)$ exactly for all $x$.

\subsection{Lower bound methods for deterministic communication complexity}

There are quite a few of lower bound methods   to determine deterministic communication complexity. We just recall so called ``rectangles" method  in this subsection. Concerning more on lower bound methods, see \cite{Buh09,Hro97,KusNis97}.

A {\em rectangle} in $X\times Y$ is a subset $R\subseteq X\times Y$ such that $R=A\times B$ for some $A\subseteq X$ and $ B\subseteq Y$. A rectangle $R=A\times B$ is called
$1(0)$-rectangle of a function $f:X\times Y\to \{0,1\}$
if for every $(x,y)\in A\times B$  the value of $f(x,y)$ is 1(0).  For  a partial
function $f: X\times Y\to \{0,1\}$ with domain $\mathcal{ D}$, a rectangle $R=A\times B$ is called
$1(0)$-rectangle if  the value of $f(x,y)$ is 1(0) for every $(x,y)\in \mathcal{ D}\cap (A\times B)$  -- we do not care about  values for $(x,y)\not\in \mathcal{ D}$.  Moreover, $C^i(f)$ is defined  as the minimum number of $i$-rectangles that partition the space of $i$-inputs (such inputs $x$ and $y$ that $f(x,y)=i$) of $f$.

We  now recall a lemma on ``rectangles" method  from \cite{KusNis97}:
\begin{Lm}\label{lm-d-lowbound}
For every (partial) function $f$, $D(f)\geq \max\{\log{ C^1(f)},\log{ C^0(f)}\}$.
\end{Lm}

\subsection{Measure-once one-way finite automata with quantum and classical states}

In this subsection we recall the definition of 1QCFA.   Concerning more on classical and quantum automata see \cite{Gru99,Gru00,Hop79,Qiu12}.

{\em Two-way finite automata with quantum and classical states} (2QCFA) were introduced by Ambainis and Watrous \cite{Amb02} and explored also by Yakary{\i}lmaz,   Zheng and others \cite{LiF12,Yak10,Zhg12,Zhg13,Zhg13a}. Informally, a 2QCFA can be seen as a  {\em two-way deterministic finite automaton} (2DFA) with an access to a quantum memory for states of a fixed Hilbert space upon which at each step either a unitary operation is performed or a projective measurement and the outcomes of which then probabilistically determine the next move of the underlying 2DFA. 1QCFA are one-way versions of 2QCFA \cite{ZhgQiu112}. In this paper, we only use 1QCFA in which a unitary transformation is applied in every step after scanning a symbol and a measurement is performed at the end of the computation. Such model is called a measure-once 1QCFA (MO-1QCFA) and corresponds to a variant of {\em measure-once quantum finite automata}, which can also be seen as a special case of {\em one-way quantum finite automata together with classical states} defined in \cite{Qiu13}.

\begin{Df}
An MO-1QCFA ${\cal A}$ is specified by a 8-tuple
\begin{equation}
{\cal A}=(Q,S,\Sigma,\Theta,\delta,|q_{0}\rangle,s_{0},Q_a),
\end{equation}
where
\begin{enumerate}
\item $Q$ is a finite set of orthonormal quantum (basis) states;
\item $S$ is a finite set of classical states;
\item $\Sigma$ is a finite alphabet of input symbols and let
$\Sigma'=\Sigma\cup \{|\hspace{-1.5mm}c,\$\}$, where symbol $|\hspace{-1.5mm}c$ will be used as the left end-marker and symbol $\$$ as the right end-marker;
\item $|q_0\rangle\in Q$ is the initial quantum state;
\item $s_0$ is the initial classical state;
\item $Q_a \subseteq Q$ denotes the set of
accepting quantum basis states;
\item $\Theta$ is a quantum transition function
\begin{equation}
\Theta: S\times \Sigma'\to U({\cal H}(Q)),
\end{equation}
where $U({\cal H}(Q))$ is the set of unitary operations  on the Hilbert space generated by quantum states from $Q$;
\item $\delta$ is a classical transition function
\begin{equation}
\delta: S\times \Sigma'\to S,
\end{equation}
such that  $\delta(s,\sigma)=s'$, then the new classical state of the automaton is $s'$.
\end{enumerate}
\end{Df}
The computation of an MO-1QCFA
${\cal A}=(Q,S,\Sigma,\Theta,\delta,|q_{0}\rangle,s_{0},Q_a)$ on an input $w=\sigma_1 \cdots\sigma_n\in \Sigma^*$ starts with the string $|\hspace{-1.5mm}cw\$$ on the input tape. At the start, the tape head of the automaton is positioned on the left end-marker and the automaton begins the computation in the  initial classical state and
in the initial quantum state. After that,
in each  step, if  the classical state of the automaton is $s$, its tape head reads a symbol $\sigma$ and its quantum state is $|\psi\rangle$, then the automaton changes its quantum state to $\Theta(s,\sigma)|\psi\rangle$ and its classical state to $\delta(s,\sigma)$. At the end of the computation,   the projective
measurement $\{P_{a}, P_{r}\}$ is applied on the current quantum state, where $P_{a}=\sum_{|i\rangle\in Q_a}|i\rangle\langle i|$ and $P_r=I-P_a$. If the classical outcome is $a$ ($r$), then the input is accepted (rejected).

 For any state $s$, any string $w\in (\Sigma')^*$ and any $\sigma\in \Sigma$,
let $\widehat{\delta}(s,\sigma
w)=\widehat{\delta}(\delta(s,\sigma),w)$; if $|w|=0$,
$\widehat{\delta}(s,w)=s$. Let $\sigma_0=|\hspace{-1.5mm}c$ and $\sigma_{n+1}=\$$.  The probability that the automaton ${\cal A}$ accepts the input $w$ is
\begin{equation}
Pr[{\cal A}\ \text{accepts}\  w] =\|P_{a}\Theta(s_{n+1},\sigma_{n+1})\cdots \Theta(s_1,\sigma_1)\Theta(s_0,\sigma_0)|q_0\rangle\|^2,
\end{equation}
where $s_{i+1}=\widehat{\delta}(s_0,\sigma_0\cdots\sigma_i)$.
The probability that ${\cal A}$ rejects the input $w$ is $Pr[{\cal A}\  \text{rejects}\  w]=1-Pr[{\cal A}\  \text{accepts}\  w]$.

The language acceptance is a special case of so called promise problem solving.
A {\em promise problem} \cite{Gh06} over an alphabet $\Sigma$ is a pair $A = (A_{yes}, A_{no})$, where $A_{yes}$, $A_{no}\subset \Sigma^*$
are disjoint sets. Languages over an alphabet $\Sigma$ may be viewed as promise problems that obey the additional constraint
$A_{yes}\cup A_{no}=\Sigma^*$.

 A promise problem $A = (A_{yes}, A_{no})$ is solved  exactly by a finite automaton ${\cal A}$  if
\begin{itemize}
\item $\forall w\in A_{yes}$, $Pr[{\cal A}\  \text{accepts}\  w]=1$, and
\item $\forall w\in  A_{no}$, $Pr[{\cal A}\ \text{rejects}\  w]=1$.
\end{itemize}

On the other side, a finite automaton ${\cal A}$ is said to solve a promise problem $A = (A_{yes}, A_{no})$ with a one-sided error $\varepsilon$ ( $0<\varepsilon\leq\frac{1}{2}$) if
\begin{itemize}
\item $\forall w\in A_{yes}$, $Pr[{\cal A}\  \text{accepts}\  w]=1$, and
\item $\forall w\in  A_{no}$, $Pr[{\cal A}\ \text{rejects}\  w]\geq 1-\varepsilon$.
\end{itemize}

\section{Generalizations of the distributed Deutsch-Jozsa promise problem}\label{section3}

We will explore communication complexity of several generalizations of  the distributed Deutsch-Jozsa promise  problem.

\begin{Th}\label{Th-EQ_k}
 $Q_E(\text{EQ}_k)\in {\bf O}(\log n)$ for any fixed $k\geq \frac{n}{2}$.
\end{Th}
\begin{proof}
Assume that Alice is given an input $x=x_1\cdots x_n$ and Bob an input $y=y_1\cdots y_n$. The following quantum communication protocol ${\cal P}$ computes $\text{EQ}_k(x,y)$ using $n+1$ quantum basis  states $|0\rangle,|1\rangle,\ldots,|n\rangle$ as follows:
 \begin{enumerate}
  \item Alice begins with the initial quantum state $|0\rangle$ and  performs on it the unitary map $U_k$ such that $U_k|0\rangle=\sqrt{\frac{2k-n}{2k}}|0\rangle+\sqrt{\frac{n}{2k}}|1\rangle$,
  where
\begin{align}
U_k=\left(
      \begin{array}{ccc}
        \sqrt{\frac{2k-n}{2k}} & -\sqrt{\frac{n}{2k}} & \mathbf{0}  \\
        \sqrt{\frac{n}{2k}} & \sqrt{\frac{2k-n}{2k}} & \mathbf{0}  \\
        \mathbf{0} & \mathbf{0} &  \mathbf{I}_{{n-1},{n-1}}\\
      \end{array}
    \right).
\end{align}
  \item  Alice then performs the unitary map $U_h$ on her quantum state such that  $U_h|0\rangle=|0\rangle$ and  $U_h|1\rangle=\frac{1}{\sqrt{n}}\sum_{i=1}^n|i\rangle$, i.e.
  the first column of $U_h$ is $(1,0,\ldots,0)^T$,  the second column of $U_h$ is $(0,\frac{1}{\sqrt{n}},\ldots,\frac{1}{\sqrt{n}})^T$, and the other entries are arbitrary, but such that the
resulting matrix is unitary what is clearly always possible.

  \item Alice then applies to the current state the unitary map $U_x$  such that  $U_x|0\rangle=|0\rangle$ and $U_x|i\rangle=(-1)^{x_i}|i\rangle$ for $i>0$.

  \item Afterwards Alice  sends her current quantum state $|\psi_4\rangle=U_xU_hU_k|0\rangle=\sqrt{\frac{2k-n}{2k}}|0\rangle+{\sqrt{\frac{n}{2k}}\sqrt{\frac{1}{n}}} \sum_{i=1}^n(-1)^{x_i}|i\rangle$ to Bob.
  \item Bob then applies to the state that he has received  the unitary map $U_y$ such that $U_y|0\rangle=|0\rangle$ and  $U_y|i\rangle=(-1)^{y_i}|i\rangle$ for $i>0$.
  \item Bob applies the unitary map $U_k^{-1}U^{-1}_h$ to his quantum state.
 \item Afterwards Bob  measures the resulting state in the standard basis and outputs 1 if the measurement outcome is $|0\rangle$ and outputs 0 otherwise.
\end{enumerate}

The state after the step 5 will be
 \begin{align}
 |\psi_5\rangle=U_yU_xU_hU_k|0\rangle=\sqrt{\frac{2k-n}{2k}}|0\rangle+\sqrt{\frac{n}{2k}}\sqrt{\frac{1}{n}}\sum_{i=1}^n(-1)^{x_i+y_i}|i\rangle.
 \end{align}

 Therefore, if  $x=y$, then the state after the step 6 will be
 \begin{align}
 |\psi_6\rangle=U^{-1}_kU^{-1}_hU_yU_xU_hU_k|0\rangle=U^{-1}_kU^{-1}_hU_hU_k|0\rangle=|0\rangle.
 \end{align}

 If $x\neq y$, then $H(x,y)=k$ and the state after the step 6 is
  \begin{align}
 |\psi_6\rangle&=U^{-1}_kU^{-1}_hU_yU_xU_hU_k|0\rangle=U^{-1}_kU^{-1}_h   \left(\sqrt{\frac{2k-n}{2k}}|0\rangle+\sqrt{\frac{n}{2k}}\sqrt{\frac{1}{n}}\sum_{i=1}^n(-1)^{x_i+y_i}|i\rangle\right)\\
 &=U^{-1}_k  \left(\sqrt{\frac{2k-n}{2k}}|0\rangle+\sqrt{\frac{n}{2k}}\frac{1}{n}\sum_{i=1}^n(-1)^{x_i+y_i}|1\rangle+\sum_{i=2}^n\alpha_i|i\rangle\right)\\
 &=U^{-1}_k  \left(\sqrt{\frac{2k-n}{2k}}|0\rangle+\sqrt{\frac{n}{2k}}\frac{n-2k}{n}|1\rangle+\sum_{i=2}^n\alpha_i|i\rangle\right)\\
 &=\left(\sqrt{\frac{2k-n}{2k}}\sqrt{\frac{2k-n}{2k}}+\sqrt{\frac{n}{2k}}\sqrt{\frac{n}{2k}}\frac{n-2k}{n}\right)|0\rangle+\sum_{i=1}^n\beta_i|i\rangle\\
 &=\sum_{i=1}^n\beta_i|i\rangle,
 \end{align}
where $\alpha_i, \beta_i$ are  amplitudes that we do not need to  be specified more exactly.

 Because the amplitude of $|0\rangle$ is 0,  we can get the exact result after the measurement in the step 7.

 It is clear that this protocol  communicates only $\lceil\log (n+1)\rceil$ qubits.
\end{proof}

Obviously, $D(\text{EQ}_k)\leq n-k+1$.
For the case that $k=\frac{n}{2}$ and $k$ is even, $\text{EQ}_k=\text{EQ}'$ and $D(\text{EQ}_k)\in {\bf \Omega}(n)$ \cite{Buh98,Buh09}.
For the cases that   $\frac{1}{2}n\leq k<(1-\lambda) n$, where $0< \lambda<\frac{1}{2}$ is given,  we can prove, using a similar proof method as in \cite{Buh98,Buh09}, the following theorem:


 \begin{Th}\label{th2}
 Suppose $0< \lambda<\frac{1}{2}$ is given and $k$ is an even. Then   $D(\text{EQ}_k)\in {\bf \Omega}(n)$ for all $k$ such that  $\frac{1}{2}n\leq k<(1-\lambda) n$.
 \end{Th}

\begin{proof}
In order to prove the theorem, we introduce a lemma (Theorem 1 in \cite{Fr87}) first.

For $x,y\in\{0,1\}^n$, let us denote $|x\wedge y|=\sum_{i=1}^n x_i\wedge y_i$. Let also $M(n,l)$ denote the maximum of the sets cardinality $|F|$, where $F\subset \{0,1\}^n$ subject to the constraint: $|x\wedge y|\neq l$ holds for all distinct $x,y\in F$.

\begin{Lm}{\cite{Fr87}}\label{lm-fb-intersection}
If $0<\eta<\frac{1}{4}$ is given, then there exists a positive constant $\varepsilon_0=\varepsilon_0(\eta)$ such that $M(n,l)\leq (2-\varepsilon_0)^n$  for all $l$ such that $\eta n<l<(\frac{1}{2}-\eta)n$.
\end{Lm}

Let ${\cal P}$ be a deterministic protocol for $\text{EQ}_k$.
Let us  consider the set $E=\{(x,x) \,|\, W(x)=\lfloor\frac{n}{2}\rfloor\}$. For every $(x,x)\in E$, we have ${\cal P}(x,x)=1$. Suppose now that there is a 1-monochromatic rectangle $R=A\times B\subseteq \{0,1\}^n\times\{0,1\}^n$ such that ${\cal P}(x,y)=1$ for every promise pair $(x,y)\in R$. Let $S=R\cap E$. We now prove that for any distinct $(x,x),(y,y)\in S$, $|x\wedge y|\neq \lfloor\frac{n-k}{2}\rfloor$.

If $|x\wedge y|= \lfloor\frac{n-k}{2}\rfloor$, then $H(x,y)=2(\lfloor\frac{n}{2}\rfloor-\lfloor\frac{n-k}{2}\rfloor)=k$ and ${\cal P}(x,y)=0$. Since $(x,x)\in R$ and $(y,y)\in R$, we have $(x,y)\in R$ and  ${\cal P}(x,y)=0$, which is a contradiction.

Because of the assumption, we have $\frac{\lambda}{2} n<\lfloor\frac{n-k}{2}\rfloor\leq \frac{1}{4}n<(\frac{1}{2}-\frac{\lambda}{2})n$. Let $\eta=\frac{\lambda}{2}$. According to  Lemma  \ref{lm-fb-intersection},  there exists a constant $\varepsilon_0$ such that  $|S|\leq (2-\varepsilon_0)^n$.

Let us now continue the proof of  Theorem \ref{th2} . The minimum number of 1-monochromatic rectangles that partition the space of inputs is
 \begin{align}
     C^1(\text{EQ}_k)\geq\frac{|E|}{|S|}\geq \frac{{n\choose \lfloor n/2\rfloor}}{(2-\varepsilon_0)^n}>\frac{2^n/n}{(2-\varepsilon_0)^n}.
\end{align}
According to Lemma \ref{lm-d-lowbound}, the deterministic communication complexity of the problem  $EQ_k$ then holds:
 \begin{align}
D(\text{EQ}_k)\geq \log{C^1(\text{EQ}_k)}>\log{\frac{2^n/n}{(2-\varepsilon_0)^n}}=n- \log n-n\log(2-\varepsilon_0).
 \end{align}
Since $1-u\leq e^{-u}\leq 2^{-u}$, for any real number $u>0$,  we have $\log(2-\varepsilon_0)=1+\log(1-\varepsilon_0/2)<1-\varepsilon_0/2$. Therefore
 \begin{align}
D(\text{EQ}_k)\geq n-\log n-n(1-\frac{\varepsilon_0}{2})=\frac{\varepsilon_0}{2}n-\log n.
 \end{align}

 Thus, $D(\text{EQ}_k)\in {\bf \Omega}(n)$.
\end{proof}

\begin{Rm}\label{rm-odd-k}
If $k$ is odd, we can prove that $D(\text{EQ}_k)\in {\bf O}(1)$ as follows:
\begin{enumerate}
  \item [1.]Alice calculates $W(x)$ and then  sends one bit information of $W(x)$'s parity to Bob (for example, Alice sends ``1" if $W(x)$ is even and ``0" otherwise).
  \item [2.]After receiving Alice's information,  Bob   calculates $W(y)$. If the parities of $W(y)$ and $W(x)$ are the same, then $\text{EQ}_k(x,y)=1$; otherwise,  $\text{EQ}_k(x,y)=0$.
\end{enumerate}

The above protocol computes $\text{EQ}_k$ since if $H(x,y)=0$, $W(x)+W(y)$ must be even; if $H(x,y)=k$, then the parity of  $W(x)+W(y)$ must be the same as the parity of $k$.

\end{Rm}

We can now explore also the exact quantum query complexity of $\text{DJ}_k$.
\begin{Th}
The exact quantum query complexity $QT_E(\text{DJ}_k)=1$ for any fixed $k\geq \frac{n}{2}$.
\end{Th}
\begin{proof}
 Let us consider a query algorithm ${\cal A}$ that will solve the  promise problem $\text{DJ}_k$ using $n+1$  quantum basis  states $|0\rangle,|1\rangle,\ldots,|n\rangle$  and works as follows: (where the unitary transformations $U_k$ and $U_h$ are the same ones as  in the proof of the Theorem \ref{Th-EQ_k}.)
 \begin{enumerate}
 \item ${\cal A}$  begins in the state $|0\rangle$ and performs on it the unitary transformation $U_1=U_hU_k$.
 \item ${\cal A}$  performs a query $Q$.
 \item ${\cal A}$  performs the  unitary transformation $U_2=U_k^{-1}U_h^{-1}$.
  \item ${\cal A}$  measures the resulting state in the standard basis and outputs 1 if the measurement outcome is $|0\rangle$ and outputs 0 otherwise.
 \end{enumerate}
 The rest of the proof is similar to that of  Theorem \ref{Th-EQ_k}.
\end{proof}

Obviously, the exact classical query complexity of $\text{DJ}_k$ is $n-k+1$.

\section{Communication complexity of a  promise version of the disjointness problem}\label{section4}
It may seem that if we consider $ \text{DISJ}_{k}'$ as a similar
 promise version to the problem $\text{DISJ}$ as we did with $\text{EQ}_k$, we  get a similar result.

 However, the reality is a bit different. Indeed, let us denote
 \begin{equation}
    \text{DISJ}_{k}'(x,y)=\left\{\begin{array}{ll}
                    1 &\ \text{if}\  \sum_{i=1}^n x_i\wedge y_i=0 \\
                    0 &\ \text{if}\ \sum_{i=1}^n x_i\wedge y_i=k,
                  \end{array}
 \right.
\end{equation}
where  $ k\geq \frac{n}{2}$ is fixed. Using an analogous proof method as in Section 3, we can prove that $Q_E(\text{DISJ}_{k}')\in {\bf O}(\log n)$. But, when comparing to the deterministic communication complexity, this is no improvement at all. Actually, we can prove that for $k>\frac{n}{2}$, $D(\text{DISJ}_{k}')\in {\bf O}(1)$.
Indeed, let us consider the following protocol:
\begin{enumerate}
  \item [1.]Alice calculates $W(x)$. If $W(x)<k$,  Alice sends 1 as the outcome of  $\text{DISJ}_{k}'(x,y)$ to Bob; otherwise, she sends 0 to Bob.
  \item [2.]After receiving Alice's information, if Bob did not get 1 as the result of $\text{DISJ}_{k}'(x,y)$ from Alice, he then calculates $W(y)$. If $W(y)<k$, then Bob outputs  1 as the result of $\text{DISJ}_{k}'(x,y)$; otherwise, $\text{DISJ}_{k}'(x,y)=0$.
\end{enumerate}

For the case  $k=\frac{n}{2}$, we can prove that $D(\text{DISJ}_{k}')\in {\bf O}(1)$ using the following protocol:
\begin{enumerate}
  \item [1.]Alice calculates $W(x)$. If $W(x)<\frac{n}{2}$, then Alice sends 1 as the outcome of  $\text{DISJ}_{k}'(x,y)$ to Bob; if $W(x)=\frac{n}{2}$, Alice sends 0 and $x_1$ to Bob; otherwise, she sends 0 to Bob.
  \item [2.]After receiving Alice's information, if Bob did not get 1 as the result of $\text{DISJ}_{k}'(x,y)$ from Alice, he then calculates $W(y)$. If $W(y)<\frac{n}{2}$, then Bob outputs the result 1 as the of $\text{DISJ}_{k}'(x,y)$.  If $W(y)=\frac{n}{2}=W(x)$, Bob compares $y_1$ with $x_1$ and then outputs the result $\text{DISJ}_{k}'(x,y)=0$ if $y_1=x_1$ and $\text{DISJ}_{k}'(x,y)=1$ if $y_1\neq x_1$. Otherwise, $\text{DISJ}_{k}'(x,y)=0$.
\end{enumerate}

Obviously, the above protocol computes $\text{DISJ}_{k}'(x,y)$ and uses for communication   only ${\bf O}(1)$ bits.


\subsection{Quantum protocol}

Let us now explore how much of advantages can be obtained when quantum resources can be used for dealing with such communication problems as  $\text{DISJ}_{\lambda}$.
We give at first a quantum communication protocol for $\text{DISJ}_{\frac{1}{4}}(x,y)$.   From this protocol we can get the following result.

\begin{Th}\label{th1}
 $Q(\text{DISJ}_{\frac{1}{4}})\leq 3+2\log n$.
\end{Th}
\begin{proof}
Assume that Alice is given an input $x=x_1\cdots x_n$ and Bob an input $y=y_1\cdots y_n$. The quantum communication protocol ${\cal P}$ which computes $\text{DISJ}_{\frac{1}{4}}$ using $2n$ quantum basis states $\{|i,j\rangle:1\leq i\leq n, 0\leq j\leq 1 \}$ (the basis state $|i,j\rangle$ is a $2n$-dimensional column vector with the $(nj+i)$-th entry being $1$ and others being $0$'s.)  will work as follows:
\begin{enumerate}
  \item [1.] Alice starts with the quantum state $|\psi_0\rangle=|1,0\rangle=(1, \overbrace{0,\cdots,0}^{2n-1})^T$ and applies to it the following unitary transformation $U_s$:
  \begin{equation}
    U_s|\psi_0\rangle=\sum_{i=1}^n\frac{1}{\sqrt{n}}|i,0\rangle=\frac{1}{\sqrt{n}}(\overbrace{1,\cdots,1}^{n}, \overbrace{0,\cdots,0}^{n})^T.
  \end{equation}
  Alice then applies the following unitary transformation $U_x$ when $x=x_1\cdots x_n$ is the input word:
  \begin{equation}
    U_x=U_{x_n}\cdots U_{x_1}
  \end{equation}
  where
   \begin{equation}
    U_{x_i}=\left\{\begin{array}{ll}
                    I, &\ \text{if}\  x_i=0 \\
                    |i,1\rangle\langle i,0|+|i,0\rangle\langle i,1| +\sum_{j\neq i}|j,0\rangle\langle j,0|+\sum_{j\neq i}|j,1\rangle\langle j,1|,&\ \text{if}\ x_i=1.
                  \end{array}
 \right.
\end{equation}

$U_x$ is therefore a unitary transformation that exchanges the amplitudes of $|i,0\rangle$ and $|i,1\rangle$ if $x_i=1$.
  The resulting quantum state, after performing $U_x$, will be
    \begin{equation}
    |\psi_1\rangle=\frac{1}{\sqrt{n}}\sum_{i=1}^n\left((1-x_i)|i,0\rangle+x_i|i,1\rangle\right)=\frac{1}{\sqrt{n}}(\overline{x}_1,\cdots,\overline{x}_n,x_1,\cdots,x_n)^T,
  \end{equation}
    where $\overline{x}_i=1-x_i$.

  Alice then sends the resulting  quantum state $|\psi_1\rangle$ to Bob.
  \item[2.] Bob  applies to the state received the  unitary mapping $V_y$, defined for each  $y$ as follows
   \begin{equation}
    V_y|i,0\rangle=|i,0\rangle
  \end{equation}
  and
   \begin{equation}
    V_y|i,1\rangle=(-1)^{y_i}|i,1\rangle.
  \end{equation}
  The quantum state after applying $V_y$ will therefore be
   \begin{align}
    |\psi_2\rangle=\frac{1}{\sqrt{n}}(\overline{x}_1,\cdots,\overline{x}_n,(-1)^{y_1}x_1,\cdots,(-1)^{y_n}x_n)^T.
  \end{align}
  If $x_i=y_i=1$, then $(-1)^{y_i}x_i=-1=(-1)^{x_i\wedge y_i}$; if $x_i=1$ and $y_i=0$, then $(-1)^{y_i}x_i=1=(-1)^{x_i\wedge y_i}$; otherwise  $(-1)^{y_i}x_i=0$.

Bob then sends his quantum state $ |\psi_2\rangle$ to Alice.
  \item [3.] Alice applies the unitary transformation $U_x$ to the  state $|\psi_2\rangle$ received from Bob and gets a new quantum state:
    \begin{equation}
    |\psi_{3}\rangle=\frac{1}{\sqrt{n}}( z_1,\cdots,z_n,\overbrace{0,\cdots,0}^{n})^T.
  \end{equation}
  If $x_i=0$, then $z_i=\overline{x}_i=1=(-1)^{x_i\wedge y_i}$. If $x_i=1$, then $z_i=(-1)^{y_i}x_i=(-1)^{x_i\wedge y_i}$.
  Therefore, $z_i=(-1)^{x_i\wedge y_i}$ for $1\leq i\leq n$.

  Alice then applies the unitary transformation $U_f$ (to  be specified later)  to get the following state:
   \begin{equation}
    U_f |\psi_{3}\rangle=\left(  \frac{1}{n}\sum_{i=1}^n (-1)^{x_i\wedge y_i}, \overbrace{ *,\cdots, *}^{2n-1}  \right)^T.
  \end{equation}
and then she measures the resulting quantum state with the observable $\{|i,0\rangle\langle i,0|,\linebreak[0]|i,1\rangle\langle i,1| \}_{i=1}^n$. If the measurement outcome is $|1,0\rangle$,  Alice sends 1 otherwise 0 to Bob.
\end{enumerate}

It is clear that this protocol uses for communication  $1+2(\log{2n})=3+2\log n$ qubits.   Unitary transformations $U_{s}$ and $U_{f}$ do exist. The first column of $U_{s}$ is $\frac{1}{\sqrt{n}}(\overbrace{1,\cdots,1}^{n},\overbrace{0,\cdots,0}^n)^T$ and the first row of $U_f$ is $\frac{1}{\sqrt{n}}(\overbrace{1,\cdots,1}^{n},\overbrace{0,\cdots,0}^n)$.
It is easy to verify that $V_y$'s are unitary
transformations.

If $\sum_{i=1}^n x_i\wedge y_i=0$, then $ \frac{1}{n}\sum_{i=1}^n (-1)^{x_i\wedge y_i}=1$. After the measurement, Alice gets the quantum outcome $|1,0\rangle$ and sends 1 to Bob. Thus,
 \begin{equation}
    Pr({\cal P}(x,y)=\text{DISJ}_{\frac{1}{4}}(x,y))=1.
  \end{equation}

If $n/4\leq \sum_{i=1}^n x_i\wedge y_i\leq 3n/4$, then $|\frac{1}{n}\sum_{i=1}^n (-1)^{x_i\wedge y_i}|\leq 1/2$ and Alice gets as  the quantum outcome
$|1,0\rangle$ with the probability not more than $|\frac{1}{n}\sum_{i=1}^n (-1)^{x_i\wedge y_i}|^2=1/4$. Thus,
 \begin{equation}
    Pr({\cal P}(x,y)=\text{DISJ}_{\frac{1}{4}}(x,y))=1-\left|\frac{1}{n}\sum_{i=1}^n (-1)^{x_i\wedge y_i}\right|^2 \geq \frac{3}{4}.
  \end{equation}
  Therefore ${\cal P}$ is a bounded error  protocol for $\text{DISJ}_{\frac{1}{4}}$ and $Q(\text{DISJ}_{\frac{1}{4}})\leq 3+2\log n$.
\end{proof}

Now we are in position to deal with the
general case.

\begin{Th}\label{q-g}
$Q(\text{DISJ}_{\lambda})\leq \frac{\log 3}{3\lambda}(3+2\log n)$, where   $0<\lambda\leq \frac{1}{4}$.
\end{Th}

\begin{proof}
For the general case,  the new quantum protocol  ${\cal P}'$ works as follows:
Repeat the protocol ${\cal P}$ from the proof of  previous theorem  $k$ times ($k$ will be specified later). If all  measurement outcomes in Step 3 are $|1,0\rangle$, then ${\cal P}'(x,y)=1$; otherwise, ${\cal P}'(x,y)=0$.

If $\sum_{i=1}^n x_i\wedge y_i=0$, then
\begin{equation}
    Pr({\cal P}(x,y)=1)=1
\end{equation}
 and
 \begin{equation}
    Pr({\cal P}(x,y)=0)=0.
\end{equation}
Therefore,
 \begin{align}
    Pr({\cal P}'(x,y)&=\text{DISJ}_{\lambda}(x,y)=1)=1.
\end{align}

If $\lambda n\leq \sum_{i=1}^n x_i\wedge y_i\leq (1-\lambda)n$, then
  \begin{align}
    p_0&=Pr({\cal P}(x,y)=\text{DISJ}_{\lambda}(x,y)=0)=1-|\frac{1}{n}\sum_{i=1}^n (-1)^{x_i\wedge y_i}|^2 \geq 1-|1-2\lambda|^2\\
    &=4\lambda-\lambda^2=4\lambda(1-\lambda)\geq 4\lambda(1-\frac{1}{4})=3\lambda.
\end{align}
 If $k=\frac{\log 1/3}{\log (1-3\lambda)}$, and the protocol ${\cal P}$ is repeated  $k$ times, then
 \begin{align}
    Pr({\cal P}'(x,y)&=\text{DISJ}_{\lambda}(x,y)=0)=1-(1-p_0)^k\geq 1-(1-3\lambda)^k\geq 1-(1-3\lambda)^{\frac{\log 1/3}{\log (1-3\lambda)}}\\
    &=1-2^{\log ((1-3\lambda)^{\frac{\log 1/3}{\log (1-3\lambda)}})}=1-2^{\frac{\log 1/3}{\log (1-3\lambda)}\times \log ((1-3\lambda)}=1-2^{\log{1/3}}=\frac{2}{3}
    .
\end{align}
 Since $1-u\leq e^{-u}\leq 2^{-u}$, for any real number $u>0$,  we have
 \begin{align}
  k=\frac{\log 1/3}{\log (1-3\lambda)}\leq \frac{\log 1/3}{\log 2^{(-3\lambda)}}=\frac{\log 3}{3\lambda}.
\end{align}
Thus, $Q(\text{DISJ}_{\lambda})\leq \frac{\log 3}{3\lambda}(3+2\log n)$.
\end{proof}

\subsection{Deterministic lower bound}

To prove the  main result, we will use a  modification of the  lower bound proof method from \cite{Buh98,Buh09}.
\begin{Th}\label{th3}
 $D(\text{DISJ}_{\lambda})\in {\bf \Omega}(n)$, where   $0<\lambda\leq \frac{1}{4}$.
\end{Th}
\begin{proof}
Let ${\cal P}$ be a deterministic protocol for $\text{DISJ}_{\lambda}$. Let us consider the set $F_{\lambda}=\{x\in\{0,1\}^n \,|\, \lambda n\leq   W(x)\leq (1-\lambda)n\}$.
If $x\in F_{\lambda}$, then also $\overline{x}\in F_{\lambda}$, where $\overline{x}=\overline{x}_1\cdots \overline{x}_n$.
Let  $E=\{(x,\overline{x}) \,|\, x\in F_{\lambda}\}$. For every $(x,\overline{x})\in E$, we then have ${\cal P}(x,\overline{x})=1$.
Suppose  now that there is a 1-monochromatic rectangle $R=A\times B\subseteq \{0,1\}^n\times\{0,1\}^n$ such that   ${\cal P}(x,y)=1$ for every pair of promise input $(x,y)\in R$. For $S=R\cap E$,  we now prove that $|S|<1.99^n$.

Suppose $|S|\geq 1.99^n$. According to Corollary 1.2 from \cite{Fr87},   there exist $(x,\overline{x})\in S$ and $(z,\overline{z})\in S$ such that $|x\wedge z|=\frac{n}{4}$. Since $S\subseteq E$, we have $x, \overline{x},z,\overline{z}\in F_{\lambda}$. Without a lost of generality, let
\begin{align}
    x&=\overbrace{1\cdots 1}^{n/4}\ \overbrace{ 0\cdots 0}^{\lambda n}\  \overbrace{ 1\cdots 1}^{\lambda n}\ \overbrace{ *\cdots *}^{3n/4-2\lambda n}\ \text{and}\\
    z&=\overbrace{1\cdots 1}^{n/4}\ \overbrace{ 1\cdots 1}^{\lambda n}\  \overbrace{ 0\cdots 0}^{\lambda n}\ \overbrace{ *\cdots *}^{3n/4-2\lambda n}
\end{align}
such that  $|x\wedge z|=\frac{n}{4}$.
In such a case
\begin{align}
    \overline{x}=\overbrace{0\cdots 0}^{n/4}\ \overbrace{ 1\cdots 1}^{\lambda n}\  \overbrace{ 0\cdots 0}^{\lambda n}\ \overbrace{ *\cdots *}^{3n/4-2\lambda n}
\end{align}
and therefore $\lambda n \leq |z\wedge \overline{x}|\leq 3n/4-\lambda n <(1-\lambda) n$. Thus, ${\cal P}(z,\overline{x})=0$. Since $S\subset R$ and $R$ is a 1-rectangle, we get $(x,\overline{x})\in R, (z,\overline{z})\in R$ and also $(z,\overline{x})\in R$. Since $(z,\overline{x})$ is a pair of the promise input, it holds ${\cal P}(z,\overline{x})=1$,  which is a contradiction.

Therefore, the minimum number of 1-monochromatic rectangles that partition the space of inputs is
 \begin{align}
     C^1(\text{DISJ}_{\lambda})\geq \frac{|E|}{|S|}= \frac{|F_{\lambda}|}{|S|}\geq \frac{|F_{1/4}|}{|S|}>\frac{2^n/2}{1.99^n}.
\end{align}
According to Lemma \ref{lm-d-lowbound}, the deterministic communication complexity then holds:
 \begin{align}
D(\text{DISJ}_{\lambda})\geq \log{C^1(\text{DISJ}_{\lambda})}>\log{(\frac{2^n/2}{1.99^n})}
=n-1-n\log{1.99}&\\>n-1-0.9927n=0.0073n-1.
 \end{align}
 Thus, $D(\text{DISJ}_{\lambda})\in {\bf \Omega}(n)$.
\end{proof}

\begin{Rm}
The lower bound  proved in the previous theorem is quite a weak bound. We expect that a better lower bound will be relative to $\lambda$. When $\lambda$ is close to 0, then the lower bound is expected to be close to $n$ instead of 0.007n.
\end{Rm}

\subsection{Probabilistic protocol}
As already mentioned,  the two-sided error  probabilistic communication complexity  $R(\text{DISJ})\linebreak[0] \in{\bf \Omega}(n)$.  However, for $\text{DISJ}_{\lambda}$, the  communication complexity can be dramatically improved as will now be shown.

Let us first deal with the case $\lambda=\frac{1}{4}$.
\begin{Th}\label{Th-1PFA}
$R(\text{DISJ}_{\frac{1}{4}})\leq 5\log{n}$.
\end{Th}
\begin{proof}
Let us consider the probabilistic protocol ${\cal P}$ which works as follows (where integer $k$ will be
speified later).
\begin{enumerate}
  \item If $W(x)<k$, then Alice  sends 1 as the result of $\text{DISJ}_{\frac{1}{4}}(x,y)$ to Bob.  Otherwise, Alice chooses randomly $k$  1's of her input, says $x_{i_1},\cdots, x_{i_k}$,  and sends their positions $i_1,\cdots, i_k$ to Bob.
  \item If Bob does not receive 1 as the result from Alice, then he checks the positions $i_1,\cdots, i_k$ of his input. If there exists a  $1\leq j\leq k$ such that $y_{i_j}=1$ , then ${\cal P}(x,y)=0$; otherwise ${\cal P}(x,y)=1$.
\end{enumerate}

If $\sum_{i=1}^n x_i\wedge y_i=0$, then
\begin{equation}
    Pr({\cal P}(x,y)=\text{DISJ}_{\frac{1}{4}}(x,y)=1)=1.
\end{equation}

If $n/4\leq \sum_{i=1}^n x_i\wedge y_i\leq 3n/4$, then for any $i\in\{i_1,\cdots,i_k\}$
\begin{equation}
    Pr(y_{i}=x_i)\geq \frac{1}{4}.
\end{equation}
Therefore,
\begin{equation}
   Pr({\cal P}(x,y)=0)\geq 1-(1-{\frac{1}{4}})^k=1-(\frac{3}{4})^k.
\end{equation}
If $k=5$, then $Pr({\cal P}(x,y)=0)>0.76> \frac{2}{3}$. Since Alice needs $\log{n}$ bits to specifies every position, we have $R(\text{DISJ}_{\frac{1}{4}})\leq 5\log{n}$.
\end{proof}

A more general result we get for all problems $R(\text{DISJ}_{\lambda})$  where $0<\lambda\leq \frac{1}{4}$.
\begin{Th}\label{p-g}
 $R(\text{DISJ}_{\lambda})\leq \frac{\log{3}}{\lambda}\log{n}$, where   $0<\lambda\leq \frac{1}{4}$
\end{Th}
\begin{proof}
For this general cases, we will use almost the same protocol as in the proof of the previous theorem, only Alice will  send to Bob more positions of 1's in her input. It holds:

If $\sum_{i=1}^n x_i\wedge y_i=0$, then
\begin{equation}
    Pr({\cal P}(x,y)=\text{DISJ}_{\lambda}(x,y)=1)=1.
\end{equation}

If $\lambda n\leq \sum_{i=1}^n x_i\wedge y_i\leq (1-\lambda)n$,  then for any $i\in\{i_1,\cdots,i_k\}$
\begin{equation}
    Pr(y_{i}=x_i)\geq \lambda.
\end{equation}
Therefore,
\begin{equation}
   Pr({\cal P}(x,y)=0)\geq 1-(1-\lambda)^k.
\end{equation}
If $k=\frac{\log{1/3}}{\log{(1-\lambda)}}$, then $(1-\lambda)^{\frac{\log{1/3}}{\log{(1-\lambda)}}}=\frac{1}{3}$ and
$
Pr({\cal P}(x,y)=0)\geq \frac{2}{3}.
$
Thus, $R(\text{DISJ}_{\lambda})\leq \frac{\log{1/3}}{\log{(1-\lambda)}}\log{n}\leq \frac{\log{3}}{\lambda}\log{n}$.
\end{proof}

\begin{Rm}
We can also define two-sided error mode as tolerating an error  probability   $\varepsilon$ instead of $\frac{1}{3}$. Modifying our proof in  Theorem \ref{q-g} and Theorem \ref{p-g}, we can get  $Q(\text{DISJ}_{\lambda})\leq \frac{\log \varepsilon}{3\lambda}(3+2\log n)$ and $R(\text{DISJ}_{\lambda})\leq \frac{\log{\varepsilon}}{\lambda}\log{n}$ for any error  probability $\varepsilon$.
\end{Rm}

\section{Applications to quantum, probabilistic and deterministic finite automata}

It has been known, since the paper \cite{Amb98}, that for some regular languages 1QFA can be more succinct than their classical counterparts.  However,
Klauck \cite{Kla00} proved, for any regular language $L$,  that the state complexity  of the exact one-way quantum finite automata for $L$ is not less than the state complexity of an equivalent one-way deterministic finite automata (DFA).
Surprisingly, situation is again different for some promise problems  \cite{AmYa11,GQZ14b,Zhg13}.

For any $n\in {\mathbb{Z}}^+$,  let us consider the  promise problem $A_{EQ_k}(n)$ over an alphabet $\Sigma=\{0,1,\#\}$,  corresponding to the $\text{EQ}_k$ problem,  that is defined as follow:
 \begin{align}
&A_{EQ_k}(n):\left\{\begin{array}{l}
                    A_{yes}(n)=\{x\#y\,|\,H(x,y)=0,x,y\in\{0,1\}^n\} \\
                    A_{no}(n)=\{x\#y\,|\,H(x,y)=k,x,y\in\{0,1\}^n \},
                  \end{array}
 \right.
\end{align}
where $k$ is a fixed even such that $k\ge n/2$ .

The quantum protocol for $\text{EQ}_k$ which is described  in Theorem \ref{Th-EQ_k} can be implemented on an MO-1QCFA as shown bellow.
Therefore, we get the following result:

\begin{Th}
The promise problem $A_{EQ_k}(n)$ can be solved exactly by an MO-1QCFA  ${\cal A}(n)$ with $n+1$ quantum basis states and ${\bf O}(n)$ classical states, whereas the sizes of the corresponding DFA  are $2^{{\bf \Omega}(n)}$ if $k$ is an even such that $\frac{1}{2}n\leq k<(1-\lambda) n$, where $0< \lambda<\frac{1}{2}$ is given.
\end{Th}

\begin{figure}[htbp]
\begin{tabular}{|l|}
    \hline

\begin{minipage}[t]{0.93\textwidth}
\begin{enumerate}
\item[1.] Read the left end-marker $\,|\hspace{-1.2mm}c$,  perform $\Theta(s_0,\,|\hspace{-1.2mm}c)=U_{|\hspace{-1mm}c}=U_hU_k$ on the initial quantum state $|0\rangle$,  change its classical state to $\delta(s_0,\ |\hspace{-1.5mm}c )=s_1$, and move the tape head one cell to the right.

\item[2.] While the currently  scanned symbol $\sigma$ is not $\#$, do the following:
 \begin{enumerate}
 \item[2.1] Apply $\Theta(s_i,\sigma)=U_{i,\sigma}$ to the current quantum state.
 \item[2.2] Change the classical state $s_i$ to $s_{i+1}$ and move the tape head one cell to the right.
\end{enumerate}
\item[3.] Change the classical state $s_{n+1}$ to  $s_1$ and move the tape head one cell to the right.

\item[4.] While the currently  scanned symbol $\sigma$ is not the right end-marker $\$$, do the following:
 \begin{enumerate}
 \item[4.1] Apply $\Theta(s_i,\sigma)=U_{i,\sigma}$ to the current quantum state.
 \item[4.2] Change the classical state $s_i$ to $s_{i+1}$ and move the tape head one cell to the right.
\end{enumerate}

\item[5.] When the right end-marker  is reached,    perform $\Theta(s_{n+1},\$)=U_{\$}=U_k^{-1}U_h^{-1}$ on the current quantum state and
measure the current quantum state with the projective measurement $\{P_a=|0\rangle\langle0|, P_r=I-|0\rangle\langle0|\}$.   If the outcome is $|0\rangle$, accept the input; otherwise reject the input.

\end{enumerate}

\end{minipage}\\

\hline
\end{tabular}
 \centering\caption{  Description of the behavior of ${\cal A}(n)$ when solving the promise problem $A_{EQ_k}(n)$. }\label{f-EQ-k}
\end{figure}

\begin{proof}
Let $x=x_1\cdots x_n$ and $y=y_1\cdots y_n$.  Let us consider an MO-1QCFA ${\cal A}(n)=(Q,S,\Sigma,\Theta,\delta,|0\rangle,s_{0},Q_a)$, where
$Q=\{|i\rangle\}_{i=0}^n$, $S=\{s_i\}_{i=0}^{n+1}$ and $Q_a=\{|0\rangle\}$.
 ${\cal A}(n)$ will start in  the initial
quantum state $|0\rangle$ and then perform the  unitary transformation $\Theta(s_0,|\hspace{-1.5mm}c)=U_{|\hspace{-1mm}c}=U_hU_k$ to the  state $|0\rangle$, where $U_h,U_k$ are the ones defined in the proof of Theorem \ref{Th-EQ_k}. We use classical states $s_i\in S$ ($1\leq i\leq n+1$) to point out the positions of the tape head that will provide some information for  quantum transformations. If  the classical state of ${\cal A}(n)$ is $s_i$ ($1\leq i\leq n$), then  the next scanned symbol of the tape head is the $i$-th symbol of $x$($y$) and $s_{n+1}$ means that the next scanned symbol of  the tape head  is  $\#$($\$$).
 The automaton proceeds as shown in Figure \ref{f-EQ-k}, where
\begin{align}
U_{i,\sigma}|i\rangle=(-1)^{\sigma}|i\rangle \text{ and }  U_{i,\sigma}|j\rangle=|j\rangle \ \text{for}\ j\neq i
\end{align}

The rest  of the proof is   analogues to the proof in   Theorem \ref{Th-EQ_k}.

The deterministic communication complexity of  $\text{EQ}_k$ is ${\bf \Omega}(n)$. Therefore,  the sizes of the corresponding DFA are $2^{{\bf \Omega}(n)}$ \cite{KusNis97}.
\end{proof}

We now apply also to finite automata the communication complexity results for $\text{DISJ}_{\lambda}$. Let us consider the following promise problem
\begin{align}
A_{D}(n):\left\{\begin{array}{l}
                    A_{yes}(n)=\{x\#y\#x\,|\,\sum_{i=1}^n x_i\wedge y_i=0,x,y\in\{0,1\}^n\} \\
                    A_{no}(n)=\{x\#y\#x\,|\,\frac{1}{4} n\leq \sum_{i=1}^n x_i\wedge y_i\leq \frac{3}{4} n,x,y\in\{0,1\}^n\}.
                  \end{array}
 \right.
\end{align}

We implement the  protocols used in Section 4 for an MO-1QCFA  and for a one-way probabilistic finite automaton (1PFA) and get the following result:
\begin{Th}
The promise problem $A_{D}(n)$ can be solved with one-sided error $\frac{1}{4}$ by an MO-1QCFA  ${\cal A}(n)$ with $2n$ quantum basis states and ${\bf O}(n)$ classical states and also by a 1PFA ${\cal P}(n)$  with ${\bf O}(n^5)$ states,  whereas the sizes of the corresponding DFA  are $2^{{\bf \Omega}(n)}$.
\end{Th}

\begin{figure}[htbp]
\begin{tabular}{|l|}
    \hline
\begin{minipage}[t]{0.93\textwidth}
\begin{enumerate}
\item[1.] Read the left end-marker $\ |\hspace{-1.5mm}c$,  perform $U_s$ on the initial quantum state $|1,0\rangle$,  change its classical state to $\delta(s_0,\ |\hspace{-1.5mm}c )=s_1$, and move the tape head one cell to the right.

\item[2.] While the currently  scanned symbol $\sigma$ is not $\#$, do the following:
 \begin{enumerate}
 \item[2.1] Apply $\Theta(s_i,\sigma)=U_{i,\sigma}$ to the current quantum state.
 \item[2.2] Change the classical state $s_i$ to $s_{i+1}$ and move the tape head one cell to the right.
\end{enumerate}
\item[3.] Move the tape head one cell to the right.

\item[4.] While the currently  scanned symbol $\sigma$ is not $\#$, do the following:
 \begin{enumerate}
 \item[4.1] Apply $\Theta(s_{n+i},\sigma)=V_{i,\sigma}$ to the current quantum state.
 \item[4.2] Change the classical state $s_{n+i}$ to $s_{n+i+1}$ and move the tape head one cell to the right.
\end{enumerate}

\item[5.] Change the classical state $s_{2n+1}$ to  $s_1$ and move the tape head one cell to the right.

\item[6.] While the currently  scanned symbol $\sigma$ is not the right end-marker $\$$, do the following:
 \begin{enumerate}
 \item[6.1] Apply $\Theta(s_i,\sigma)=U_{i,\sigma}$ to the current quantum state.
 \item[6.2] Change the classical state $s_i$ to $s_{i+1}$ and move the tape head one cell to the right.
\end{enumerate}

\item[7.] When the right end-marker $\$$ is reached,    perform $U_{f}$ on the current quantum state,
measure the current quantum state with the projective measurement $\{P_a=|1,0\rangle\langle 1,0|, P_r=I-P_a\}$.     If the outcome is $|1,0\rangle$, accept the input; otherwise reject the input.

\end{enumerate}

\end{minipage}\\

\hline
\end{tabular}
 \centering\caption{  Description of the behavior of ${\cal A}(n)$ when solving the promise problem $A_{D}(n)$. }\label{f-A-D}
\end{figure}

\begin{proof}
Let $x=x_1\cdots x_n$ and $y=y_1\cdots y_n$.  Let us consider an MO-1QCFA ${\cal A}(n)=(Q,S,\Sigma,\Theta,\delta,|q_0\rangle,s_{0},Q_a)$, where
$Q=\{|i,0\rangle,|i,1\rangle \}_{i=1}^n$, $|q_0\rangle=|1,0\rangle$ and $Q_a=\{|1,0\rangle\}$.
 The automaton proceeds as shown in Figure \ref{f-A-D}, where $U_s,\ U_f$ are the ones defined in the proof of Theorem \ref{th1}  and
\begin{align}
& U_{i,\sigma}|j,0\rangle=|j,1\rangle \ \text{and}\ U_{i,\sigma}|j,1\rangle=|j,0\rangle \ \text{if } \sigma=1 \ \text{and}\ j=i, \ \text{otherwise}\  U_{i,\sigma}|j,k\rangle=|j,k\rangle;\\
& V_{i,\sigma}|j,1\rangle=(-1)^{\sigma}|j,1\rangle \ \text{if } j=i ,\  \text{otherwise}\  V_{i,\sigma}|j,k\rangle=|j,k\rangle;
\end{align}

It is easy to verify that for $1\leq i\leq n$,  $ U_{i,\sigma}$ and  $V_{i,\sigma}$ are unitary transformations. According to the analysis in the proof of Theorem \ref{th1},
if the input string $w\in A_{yes}(n)$,  then the automaton will get the outcome $|1,0\rangle$ in Step 7 with certainty and therefore
\begin{align}
Pr[{\cal A}\  \text{accepts}\  w]=1.
\end{align}
If the input string $w\in A_{no}(n)$,  the automaton gets the  outcome
$|1,0\rangle$ with probability not more than $1/4$. Thus,
\begin{align}
Pr[{\cal A}\  \text{rejects}\  w]\geq \frac{3}{4}.
\end{align}

Using the protocol from the proof of Theorem \ref{Th-1PFA} and the proof that its probabilistic communication complexity is not more than $5\log n$, it is easy to design a 1PFA with $O(n^5)$ states to solve the promise problem.

The deterministic state complexity lower bound can now be proved as follows.

Let  an $N$-states
DFA ${\cal A}'(n)=(S,\Sigma,\delta,s_0,S_{acc})$ solves the promise problem $A_{D}(n)$,  then we can get a deterministic protocol for $\text{DISJ}_{\frac{1}{4}}(x,y)$  as follows:
\begin{enumerate}
  \item [1.] Alice simulates the computation of ${\cal A}'(n)$ on the input ``$x\#$" and then sends her state $\widehat{\delta}(s_0,x\#)$ to Bob.
  \item [2.]Bob simulates the computation of ${\cal A}'(n)$ on the input ``$y\#$" starting at the state $\widehat{\delta}(s_0,x\#)$, and then sends his state  $\widehat{\delta}(s_0,x\#y\#)$ to Alice.
  \item [3.]Alice simulates the computation of ${\cal A}'(n)$ on the input ``$x$" starting at the state $\widehat{\delta}(s_0,x\#y\#)$. If $\widehat{\delta}(s_0,x\#y\#x)\in S_{acc}$, then Alice sends the result 1 to Bob, otherwise Alice sends the result 0 to Bob.
\end{enumerate}

The deterministic complexity of the above protocol is $ 1+2 \log{N}$ and therefore $D(\text{DISJ}_{\frac{1}{4}})\linebreak[0]\leq 1+2 \log{N}$. According to the analysis in Theorem \ref{th3}, we have
\begin{align}
&1+2 \log{N}\geq D(\text{DISJ}_{\frac{1}{4}})> 0.0073n-1\\
&\Rightarrow N\in 2^{{\bf \Omega}(n)}.
\end{align}
\end{proof}

\section{Conclusion}
We have explored generalizations of the Deutsch-Jozsa promise problem and its communication  and also  query complexities.  We have proved that the exact quantum communication complexity  $Q_E(\text{EQ}_k)\in {\bf O}(\log n)$ for any fixed $k\geq \frac{n}{2}$, whereas the  exact classical communication complexity  $D(\text{EQ}_k)\in {\bf \Omega}(n)$ if $k$ is an  even  such that $\frac{1}{2}n\leq k<(1-\lambda) n$, where $0< \lambda<\frac{1}{2}$ is given.
We have also shown that the exact quantum query complexity $QT_E(\text{DJ}_k)=1$ for any fixed $k\geq \frac{n}{2}$, whereas the exact classical query complexity $DT(\text{DJ}_k)=n-k+1$.
Promise versions of the disjointness problem  also have been discussed. We have proved that for some promise versions of the disjointness problem that   there exist   exponential gaps between  quantum (and also probabilistic) communication complexity  and  deterministic communication complexity.

Using results of the communication complexity to prove lower bounds of the state complexity of finite automata is one of the important methods \cite{Hro01,Kla00,KusNis97}.  In this paper we have used them not only  to prove lower bounds but also upper bounds.  Two communicating parties Alice and Bob are supposed to have access to arbitrary computational power in communication complexity models. However,  we have also designed communication protocols  in  Section~\ref{section3} and Section~\ref{section4} in which both Alice and Bob are  using very limited computational power. The computations of both Alice and Bob can even be simulated by  finite automata.

Some  problems for future work are:
\begin{enumerate}
  \item  We have  generalized the  distributed Deutsch-Jozsa promise problem to determine whether $H(x,y)=0$  or $H(x,y)= k$, where $k$ is a fixed integer such that $k\geq \frac{n}{2}$. Does there exist similar results  for some cases where $k<\frac{n}{2}$?
  \item  Does there exist a promise version of the   disjointness problem such that its exact quantum communication complexity can be exponential better than its deterministic communication complexity?
\end{enumerate}

\section*{Acknowledgements}
The authors are thankful to  anonymous referees   for their comments and suggestions that greatly help to improve the quality of the manuscript.
We  also thank  Abuzer Yakary{\i}lmaz for useful comments on earlier drafts of this paper.

\end{document}